\documentclass[letterpaper, 10 pt, conference]{ieeeconf} 

\usepackage{cite}
\usepackage{amssymb,amsfonts}
\usepackage{graphicx}
\usepackage{graphics}
\usepackage{tcolorbox}
\usepackage{url}
\usepackage{xcolor}
\usepackage{hyperref}
\usepackage{textcomp}
\usepackage{mathtools, cuted}
\usepackage{lipsum}
\usepackage{stfloats}
\usepackage{multirow}
\usepackage{subcaption}
\usepackage{pdfpages}
\usepackage{amssymb,euscript,psfrag,latexsym,graphicx}
\usepackage{bbm,color,amstext,wasysym,cuted,mathtools, cite}
\usepackage{nopageno}
\usepackage[normalem]{ulem}
\graphicspath{{./},{./figures/}}
\usepackage{dsfont}
\usepackage{ulem}
\usepackage{wrapfig}
\usepackage{algorithm}
\usepackage[noend]{algpseudocode}

\newcommand{\cA}{{\cal A}}
\newcommand{\cC}{{\cal C}}

\newcommand{\cF}{{\cal F}}

\newcommand{\cS}{{\cal S}}
\newcommand{\cL}{{\cal L}}
\newcommand{\cP}{{\cal P}}

\newcommand{\mC}{{\mathbb C}}

\newcommand{\mR}{{\mathbb R}}

\newcommand{\mU}{{\mathbb U}}

\newcommand{\bG}{{\mathbf G}}

\newcommand{\bs}{{\mathbf s}}
\newcommand{\bv}{{\mathbf v}}

\newcommand{\bF}{{\mathbf F}}

\newcommand{\bP}{{\mathbf P}}

\newcommand{\by}{{\mathbf y}}
\newcommand{\bff}{{\mathbf f}}

\newcommand{\bw}{{\mathbf w}}
\newcommand{\bx}{{\mathbf x}}
\newcommand{\bu}{{\mathbf u}}

\newcommand{\bA}{{\mathbf A}}

\newcommand{\bC}{{\boldsymbol{\mathcal C}}}

\newcommand{\bX}{{\mathbf X}}

\newcommand{\bz}{{\mathbf z}}

\newcommand{\bc}{{\mathbf c}}

\newtheorem{theorem}{Theorem}
\newtheorem{definition}{Definition}
\newtheorem{lemma}{Lemma}
\newtheorem{remark}{Remark}
\newtheorem{assumption}{Assumption}

\newtheorem{corollary}{Corollary}

\begin{document}

\title{\LARGE \bf Spectral Koopman Method for Identifying Stability Boundary}

\author{ Bhagyashree Umathe, and Umesh Vaidya  
}

\maketitle

\thispagestyle{empty}

\begin{abstract}
The paper is about characterizing the stability boundary of an autonomous dynamical system using the Koopman spectrum. For a dynamical system with an asymptotically stable equilibrium point, the domain of attraction constitutes a region consisting of all initial conditions attracted to the equilibrium point. The stability boundary is a separatrix region that separates the domain of attraction from the rest of the state space. For a large class of dynamical systems, this stability boundary consists of the union of stable manifolds of all the unstable equilibrium points on the stability boundary. We characterize the stable manifold in terms of the zero-level curve of the Koopman eigenfunction. A path-integral formula is proposed to compute the Koopman eigenfunction for a saddle-type equilibrium point on the stability boundary. The algorithm for identifying stability boundary based on the Koopman eigenfunction is attractive as it does not involve explicit knowledge of system dynamics. We present simulation results to verify the main results of the paper.

\end{abstract}


\section{Introduction}

Characterizing stability and identifying stability boundary of nonlinear dynamical systems is of interest to various problems in engineering and science. Examples include biological systems, biomedicine,  robotics, power systems, power electronics, and economics. Given the significance of the problem, there is a long history and a considerable body of literature on the stability theory of systems. Methods for characterizing stability and stability boundary of equilibrium dynamics can be broadly classified into two different classes: Lyapunov and non-Lyapunov-based methods \cite{chiang2015stability}.

The Lyapunov-based methods rely on constructing the Lyapunov or energy function for stability verification and providing an estimate for the stability boundary in the form of a domain of attraction. Roughly speaking, the domain of attraction is a set of all initial conditions attracted to equilibrium dynamics asymptotically \cite{Khalil_book}. One of the main challenges with the Lyapunov-based methods is that there are no systematic methods to construct these functions, and the stability boundary estimates provided are often conservative. More recent work along these lines includes the development of the maximal Lyapunov function, optimal estimation of stability region based on a given Lyapunov function, Linear Matrix Inequality, and Sum of Squares-based optimization methods for constructing Lyapunov functions \cite{Parrilothesis}. These optimization-based methods suffer from the curse of dimensionality and entail high computational efforts with the increase in the dimension of the state space. With the recent advances in computing and neural networks, there are also efforts to provide approaches based on Deep Neural Networks for constructing Lyapunov function \cite{Bevanda_Diffeomorph}.

The non-Lyapunov-based approach for stability characterization relies on exploiting the geometrical structure in the form of stable and unstable manifolds of the dynamical system for identifying stability boundary. In particular, for a large class of dynamical systems, the stability boundary can be characterized using the stable manifold of an unstable equilibrium point on the stability boundary \cite{Chiang_1988}. 

More recently, linear operator theoretic methods involving Perron-Frobenius and Koopman operators have become popular for analyzing and synthesizing nonlinear dynamical systems \cite{housparse,klus2020eigendecompositions,XuMa2019,korda2018convergence,korda2020optimal}. In \cite{Alok_Mecc}, authors have proposed uncertainty propagation using Koopman spectrum-based reachability analysis.
These operators' linearity and spectral properties are used to construct stability certificates in the form of Lyapunov functions and Lyapunov measures. Data-driven approaches are proposed for constructing the Lyapunov function using the Koopman operator. In this paper, we exploit the spectral properties of the Koopman operator to discover a non-Lyapunov-based approach for characterizing the stability boundary. This is made possible by using the fact that the spectral properties of the Koopman operator have an intimate connection with the state space geometry of the dynamical systems. In particular, a dynamical system's stable and unstable manifolds can be obtained as a zero-level curve of the eigenfunctions of the Koopman operator. We exploit this insight towards the development of systematic methods for the identification of stability boundary. There are some distinct advantages of using the Koopman spectrum to develop a non-Lyapunov approach for identifying stability boundary. We will highlight these advantages as part of the main contributions of this paper.

The main contributions of this paper are as follows. We characterize the stability boundary in terms of the principal eigenfunction of the  Koopman operator. In particular, the stability boundary is characterized by the zero-level curve of the eigenfunction corresponding to a positive eigenvalue. A path-integral formula is provided to compute the Koopman principal eigenfunction corresponding to the positive eigenvalue associated with a saddle-type unstable equilibrium point on the stability boundary. The path-integral formula is used to compute the value of eigenfunction at a few sample data points in the state space. We present convergence analysis results for the approximation of the Koopman eigenfunction based on the value of the eigenfunction at a few sample data points. One of the main advantages of the proposed Koopman-based approach for characterizing stability boundary compared to using the Koopman operator for the computation of the Lyapunov function for stability assessment is that we are only required to compute one eigenfunction. This contrasts the computation of all the dominant eigenfunctions used to construct the Lyapunov function. This is because the stability boundary forms a  manifold of co-dimension one, characterized as a zero-level curve of one eigenfunction. Another advantage unique to our proposed computational approach is that we can compute the Koopman eigenfunction directly from the data instead of obtaining it from the approximation of the Koopman operator. This is a tremendous computational advantage for systems with large dimension.

\section{Preliminaries and Notations}\label{section_preilim}
{\bf Notations}: $\mR^n$ denotes the $n$ dimensional Euclidean space. We denote $\cL_{\infty}(\bX), \cC^1(\bX)$ as the space of all essentially bounded real-valued functions and continuously differentiable functions on $\bX \subseteq \mR^n $ respectively. Here, $\bs_t(\bx)$ is the flow of $\dot \bx = \bff(\bx)$, at time $t$ with initial condition $\bx$.

In this paper, we are interested in characterizing the stability boundary using the Koopman eigenfunctions for an autonomous dynamical system of the form.
\begin{align}
    \dot \bx = \bff(\bx), \quad \bx \in \bX \subseteq \mR^n
    \label{eq:sys_dyn}
\end{align}
 where, $\bX$ is assumed to a smooth manifold without boundary or an open subset of $\mR^n$. It is assumed that the vector field is smooth enough to ensure the existence and uniqueness of the solution of the differential equation. We say an equilibrium point $\bx^\star$ is hyperbolic if, $\bA:=\frac{\partial {\bf f}}{\partial \bx}(\bx^\star)$ has no eigenvalues on the $j\omega$ axis.
The following assumption is made on the vector field (\ref{eq:sys_dyn}) in the rest of the paper.
\begin{assumption}\label{assume_system} We assume that    in system (\ref{eq:sys_dyn}),  $\bff(\bx)$ is atleast $\cC^1$ function of $\bx$ with an asymptotically stable equilibrium point at  $\bx_s$, and all the other equilibrium points hyperbolic. 
\end{assumption}

\subsection{Characterization of the Stability Boundary} 
We start with definitions related to the stability analysis. 

\begin{definition}[Stability Region and Boundary]  
For the  system (\ref{eq:sys_dyn}),
with the flow $ \bs_t(\bx)$, let $\bx_s$ be an asymptotically stable equilibrium point. We define the region $\cA(\bx_s)$ as the stability region  or the region of attraction of the stable equilibrium point $\bx_s$ as
\begin{align}
    \cA(\bx_s) := \left \{\bx \in \mR^n : \quad\lim_{t\rightarrow \infty } \bs_t(\bx) = \bx_s \right \}.
\end{align}
The boundary of $\cA(\bx_s)$ is called the stability boundary $\partial \cA(\bx_s)$, which is a separatrix of $\bx_s$.

\end{definition}

\begin{definition}\label{localmanifolds}[Local Stable and Local Unstable Manifold]  
We consider   
 $\bx^{\star}$  as a hyperbolic equilibrium point, and $D \subset \mR^n $ is a neighborhood of $\bx^{\star}$. If the flow of (\ref{eq:sys_dyn}) is $\bs_t(\bx)$ then the local stable manifold of $\bx^{\star}$ is given by:
\begin{align}
    W^s_l(\bx^{\star}) := \{ \bx \in D : \bs_t(\bx)\rightarrow \bx^{\star}\quad \text{as} \quad t\rightarrow \infty\}
\end{align}
Similarly, the local unstable manifold is given by:
\begin{align}
    W^u_l(\bx^{\star}) := \{ \bx \in D : \bs_t(\bx)\rightarrow \bx^{\star} \quad \text{as} \quad t\rightarrow -\infty\}
\end{align}
\end{definition}

\begin{definition}\label{globalmanifolds}[Stable and Unstable Manifold]\label{definiton_manifold}
    The stable manifold $W^s(\bx^{\star}) $ and Unstable manifold $W^u(\bx^{\star})$ are derived by the local manifolds with the flow $s_t(\bx)$ in backward and forward in time respectively, as:
\begin{align*}
    W^s(\bx^{\star}):=\bigcup_{t\leq 0} \bs_t \left( W^s_l(\bx^{\star})\right),\;
    W^u(\bx^{\star}):=\bigcup_{t\geq 0} \bs_t \left( W^u_l(\bx^{\star})\right).
\end{align*}
\end{definition}



For the characterization of the stability boundary, the following assumptions are made on the vector field (\ref{eq:sys_dyn}). 
\begin{assumption}\label{assume1}
\begin{enumerate}
\item [A1.] All the equilibrium points on the stability boundary are hyperbolic.
\item [A2.] The stable and unstable manifolds of the equilibrium points on the stability boundary satisfy the transversality assumption (The intersection of two manifolds at point $\bx$ is said to be transversal if the tangent space of the two manifolds at point $\bx$ span the entire space). 
\item [A3.] Every trajectory on the stability boundary approaches one of the equilibrium points at $t\to \infty$. 
\end{enumerate}
\end{assumption}
The following theorem from \cite[Theorem 4.8]{chiang2015stability}  characterizes the stability boundary under the above assumption. 

\begin{theorem}\label{theorem1}
For the nonlinear dynamical system (\ref{eq:sys_dyn}),  satisfying Assumptions \ref{assume_system} and \ref{assume1}, let $\bx_i,\;\;i=1,\ldots N$ be the equilibrium points on the stability boundary $\partial {\cal A}(\bx_s)$ of the stable equilibrium point $\bx_s$. Then
\begin{enumerate}
\item [a)] $\bx_i\in \partial {\cal A}(\bx_s)$ if and only if $W^u(\bx_i)\cap \partial {\cal A}(\bx_s)\neq \emptyset$.
\item [b)] $\partial {\cal A}(\bx_s)=\cup_{i}W^s(\bx_i)$. 
\end{enumerate}
\end{theorem}
Assumptions A1 and A2 are generic. Roughly speaking,
we say a property is generic for a class of systems if that property is true for almost all
systems in the class.  Assumption A3 is not a generic property; thus, it needs to be verified. 
The existence of the Lyapunov function provides a sufficient condition for the Assumption
A3 to hold. 
In this paper, we pursue a more direct approach for computing the stability boundary for the class of system satisfying Assumption \ref{assume1}. The proposed approach relies on directly computing the stability boundary based on constructing stable manifolds of the equilibrium points on the stability boundary. So, we study the class of system for which the Lyapunov-based approach is used to determine the domain of attraction  as the existence of the Lyapunov function is sufficient for the Assumption \ref{assume1} to be satisfied.

\begin{definition}\label{definition_typeone}
An equilibrium is said to be type-one if exactly one eigenvalue of the linearization of the system at that equilibrium point has positive real part and called source if all the eigenvalues have positive real part. 
\end{definition}
The following theorem  from  \cite[Theorem 4.10]{chiang2015stability} characterizes the structure of the equilibrium on the stability boundary. 
\begin{theorem}
For the nonlinear dynamical system 
(\ref{eq:sys_dyn}) containing two or more stable
equilibrium points, if the system satisfies Assumption \ref{assume1}, then the stability
boundary $\partial \cA(\bx_s)$ of the stable equilibrium point $\bx_s$ must contain at least one type-one equilibrium point. If, furthermore, the stability region $\cA(\bx_s)$ is bounded, then $\partial \cA(\bx_s)$ must
contain at least one type-one equilibrium point and one source.
\end{theorem}

\subsection{Koopman Operator and its Spectrum } \label{section_Koopmanspectrum}
 In this section, we provide a brief overview of existing results on the spectral theory of the Koopman operator. For more details on this topic, refer to \cite[Chapter 7]{lasota_yorke}. 

 
\begin{definition}[Koopman Operator]  $\mathbb{U}_t :{\cal L}_\infty(\bX)\to {\cal L}_\infty(\bX)$  for dynamical system~\eqref{eq:sys_dyn} is defined as 
\begin{eqnarray}[\mathbb{U}_t \psi](\bx)=\psi(\bs_t(\bx)),\;\;\;\psi(\bx)\in {\cal L}_\infty(\bX) \label{koopman_operator}.
\end{eqnarray}
The infinitesimal generator for the Koopman operator is 
\begin{eqnarray}
\lim_{t\to 0}\frac{(\mathbb{U}_t-I)\psi}{t}=\frac{\partial \psi}{\partial \bx}{\bff }(\bx)=:{\cal K}_{\bff} \psi,\;\;t\geq 0. \label{K_generator}
\end{eqnarray}
\end{definition}
\begin{definition}\label{definition_koopmanspectrum}[Eigenvalues and Eigenfunctions of Koopman] A function $\phi_\lambda(\bx)\in \cC^1(\bX)$  is said to be eigenfunction of the Koopman operator associated with eigenvalue $\lambda$ if
\begin{eqnarray}
[\mU_t \phi_\lambda](\bx)=e^{\lambda t}\phi_\lambda(\bx)\label{eig_koopman}.
\end{eqnarray}
Using the Koopman generator,  (\ref{eig_koopman}) can be written as 
\begin{align}
    \frac{\partial \phi_\lambda}{\partial \bx}{\bff}(\bx)=\lambda \phi_\lambda(\bx)\label{eig_koopmang}.
\end{align}
\end{definition}

where $\frac{\partial \phi_\lambda}{\partial \bx}=(\frac{\partial \phi_\lambda}{\partial x_1},\ldots,\frac{\partial \phi_\lambda}{\partial x_n})$ is a row vector.
The spectrum of the Koopman operator, in general, is very complex and could consist of discrete and continuous parts. In this paper, we are interested in approximating the eigenfunctions of the Koopman operator with associated eigenvalues, same as that of the linearization of the nonlinear system at the equilibrium point. With the hyperbolicity assumption on the equilibrium point of  (\ref{eq:sys_dyn}), this part of the spectrum of interest is known to be discrete and well-defined.


In the following, we briefly summarize the results from \cite{mezic2020spectrum}. Equations (\ref{eig_koopman}) and (\ref{eig_koopmang}) provide a general definition of the Koopman spectrum. However, the spectrum can be defined over finite time or a subset of the state space. 

\begin{definition}[Open Eigenfunction \cite{mezic2020spectrum}]\label{definition_openeigenfunction}
 Let $\phi_\lambda: \bC\to \mC$, where $\bC\subset \bX$ is not an invariant set. Let $\bx\in  \bC$, and
$\tau \in (\tau^-(\bx),\tau^+(\bx))= I_\bx$, a connected open interval such that $\tau (\bx) \in \bC$ for all  $\tau \in I_\bx$.
If
\[[\mU_\tau \phi_\lambda](\bx) = \phi_\lambda(\bs_\tau(\bx)) =e^{\lambda \tau}  \phi_\lambda (\bx),\;\;\;\;\forall \tau \in I_\bx, \]
then $\phi_\lambda(\bx)$ is called the open eigenfunction of the Koopman operator family $\mU_t$, for $t\in \mR$ with eigenvalue $\lambda$. 
\end{definition}
Based on the above definition, next, we define the concepts of subdomain eigenfunction, principal eigenfunctions, and Koopman spectrum.
If $\bC$ is a proper invariant subset of $\bX$, in which case $I_\bx=\mR$ for every $\bx\in \bC$, then $\phi_\lambda$ is called the subdomain eigenfunction. If $\bC=\bX$ then $\phi_\lambda$ will be the ordinary eigenfunction associated with eigenvalue $\lambda$ as defined in (\ref{eig_koopman}). 
The open eigenfunctions, as defined above, can be extended from $\bC$ to a larger reachable set when $\bC$ is open based on the construction procedure outlined in  [Definition 5.2, Lemma 5.1 \cite{mezic2020spectrum}]. Let $\cP$ be that domain. 
The eigenvalues of the linearization of the system dynamics at the origin, i.e., $\bA$, will form the eigenvalues of the Koopman operator \cite[Proposition 5.8]{mezic2020spectrum}. Our interest will be in constructing the corresponding eigenfunctions defined over the domain $\cP$. We will call these as {\it principal eigenfunctions}. 
The principal eigenfunctions will be defined over a proper subset $\cP$ of the state space $\bX$ (called subdomain eigenfunctions) or over the entire $\bX$ \cite[Lemma 5.1, Corollary 5.1, 5.2, and 5.8]{mezic2020spectrum}. 

The spectrum of the Koopman operator reveals important information about the state space geometry of the dynamical system \cite{mezic2020spectrum, mezic2021koopman}. In particular, we have the following results.

\begin{corollary}\cite[Corollary 5.10] {mezic2020spectrum}\label{proposition_mainfolds}
Let $\bx^\star$ be the hyperbolic equilibrium point of the system (\ref{eq:sys_dyn}) with $\lambda_1,\ldots, \lambda_p$ the eigenvalues of the linearization of system (\ref{eq:sys_dyn}) at $\bx^\star$. Let $\lambda_1,\ldots,\lambda_u$ be eigenvalues with positive real part with associated open eigenfunctions $\phi_{\lambda_1},\ldots, \phi_{\lambda_u}$ and $\lambda_{u+1},\ldots,\lambda_p$ be eigenvalues with negative real part with associated open eigenfunctions $\phi_{\lambda_{u+1}},\ldots, \phi_{\lambda_p}$ defined over the set $\cP$. Then, the joint level set of the eigenfunctions  
\begin{align}
W_{\cP}^s=\{\bx\in \bX: \phi_{\lambda_1}(\bx)=\ldots=\phi_{\lambda_u}(\bx)=0\},    
\end{align}
forms the stable manifold on $\cP$ and the joint level set of the eigenfunctions
\begin{align}
W_\cP^u=\{\bx\in \bX: \phi_{\lambda_{u+1}}(\bx)=\ldots=\phi_{\lambda_p}(\bx)=0\},
\end{align}
is the unstable manifold on $\cP$ of origin equilibrium point. 
\end{corollary}

The following remark on the eigenfunctions of the Koopman operator and stable, unstable subspaces corresponding to a linear system is easy to prove. 
\begin{remark}\label{remark_linearKoopman}
The Koopman eigenfunctions corresponding to the eigenvalue $\lambda_j$ of the linear dynamics, $\dot \bx=\bA \bx$, are given by $\phi_{\lambda_j}(\bx)=\bv_j^\top \bx$, where $\bv_j^\top$ is  left eigenvector of $\bA$ with eigenvalue $\lambda_j$ i.e., $\bv^\top_j\bA=\lambda_j \bv^\top_j$. 
\end{remark}

\section{Stability Boundary computation using Koopman Spectrum}\label{section_main}
This section presents the main results on the computation of stability boundary using the Koopman spectrum. 

\begin{theorem} Consider the dynamical system (\ref{eq:sys_dyn}) satisfying Assumption \ref{assume1}. Let $\bx_i$ for $i=1,\ldots, N$ be the type-one equilibrium point on the stability boundary of the stable equilibrium point $\bx_s$. Let $\phi_u^i(\bx)$ be the principal eigenfunctions corresponding to the positive eigenvalues of the linearization $\bA_i:=\frac{\partial \bff}{\partial \bx}(\bx_i)$ with positive real part and defined in the domain $\cP$.  
Then, the stability boundary can be characterized in $\cP$ using the Koopman eigenfunction as the joint zero-level set as follows
\begin{align}\cS_b:=\bigcup_{i=1}^N\{\bx: \phi^i_u(\bx)=0\}
\end{align}
\end{theorem}
\begin{proof}The proof follows by applying results of Theorem \ref{theorem1}, where it is shown that the stability boundary is formed as the union of  stable manifold of unstable equilibrium points on the boundary i.e., $\partial {\cal A}(\bx_s)=\cup_{i=1}^N W^s(\bx_i)$. From results of Corollary \ref{proposition_mainfolds}, we know that stable manifold restricted to $\cal P$ is given by $W^s_\cP(\bx_i)=\{\bx: \phi_u^i(\bx)=0\}$. The desired result is then followed by combining these two results.
\end{proof}
The above theorem provides the characterization of stability boundary in terms of Koopman eigenfunction. However, computing the Koopman eigenfunction is a challenge. We now present our main results for the computation of the Koopman eigenfunction corresponding to the eigenvalue with a positive real part. 
With no loss of generality, we will assume that the unstable equilibrium point on the stability boundary is at the origin.
Since the equilibrium point is hyperbolic (Assumption \ref{assume_system} and \ref{assume1}), the system (\ref{eq:sys_dyn}) admits the following decomposition into linear and nonlinear parts.
\begin{align}
\dot \bx=\bff(\bx)=\bA \bx+ \bF_n(\bx)\label{sys_decompose}
\end{align}
where $\bA:=\frac{\partial \bff}{\partial \bx}(0)$ and $\bF_n(\bx):=\bff(\bx)-\bA\bx$. Our goal is approximate the eigenfunction corresponding to eigenvalue $\lambda$ of $\bA$ with ${\rm Re}(\lambda)>0$. Similar to the decomposition of the system dynamics, the principal eigenfunction also admits a decomposition into linear and nonlinear parts. Let $\phi_\lambda(\bx)$ be the principal eigenfunction for eigenvalue $\lambda$, we have 
\begin{align}
\phi_\lambda(\bx)=\bw_\lambda^\top \bx+h_\lambda(\bx).\label{decompose}
\end{align}
where $\bw_\lambda^\top \bx$ is the linear part and $h_\lambda(\bx)$ is the nonlinear part.  Substituting the above expression of eigenfunction in (\ref{eig_koopmang}) and using (\ref{sys_decompose}), we obtain  $\bw_\lambda^\top \bA=\lambda\bw_\lambda^\top$
i.e., $\bw$ is the left eigenvector of $\bA$ with eigenvalue $\lambda$. We also obtain the linear partial differential equation to be satisfied by $h(\bx)$ as,
\begin{align}
\frac{\partial h_\lambda}{\partial \bx}\bff(\bx)-\lambda h_\lambda(\bx)+\bw^\top \bF_n(\bx)=0\label{pde}
\end{align}
In \cite{deka2023path}, we provided a path-integral formula for the computation of $h_\lambda(\bx)$ and hence the Koopman eigenfunction. The results in \cite{deka2023path} were developed to approximate eigenfunction, assuming the system has a stable equilibrium point. The results from \cite{deka2023path} do not apply to our case, as we are dealing with the unstable or saddle-type equilibrium point on the stability boundary. We start with the following results on the general solution formula for the linear PDE (\ref{pde}). 

\begin{theorem}\label{theorem_main} The solution formula for the first order linear PDE (\ref{pde}) can be written as 
\begin{align}
h_\lambda(\bx)=e^{-\lambda t} h_\lambda(\bs_t(\bx))+\int_0^t e^{-\lambda t} \bw^\top_\lambda \bF_n(\bs_\tau(\bx))d\tau\label{pde_soultion}
\end{align} 
\end{theorem}
\begin{proof}
The PDE (\ref{pde}) can be written as 
\begin{align}
\frac{d h_\lambda(\bs_t(\bx))}{dt}-\lambda h_\lambda(\bs_t(\bx))+\bw^\top_\lambda \bF_n(\bs_t(\bx))=0.
\end{align}
Multiplying throughout by $e^{-\lambda t}$, we obtain
$\frac{d (e^{-\lambda t} h_\lambda(\bs_t(\bx)))}{dt}+e^{-\lambda t}\bw^\top_\lambda \bff_n(\bs_t(\bx))=0$.
Next, by integrating the above from $0$ to $t$, we obtain the desired result.

\end{proof}
\noindent Before proving the main results we have following remark.

\begin{remark} We are not interested in computing the eigenfunction per se but the zero-level curve of the eigenfunction characterizing the stability boundary. Hence, our interest will be in approximating the Koopman eigenfunction in the region containing the stable manifold. Following the Definitions \ref{localmanifolds} and \ref{globalmanifolds} of local and global stable manifolds,  we notice that the region containing the stable manifold can be obtained by backward propagating a small set around the origin containing the local stable manifold. This observation is crucial to approximate eigenfunction, as the following theorem proves. 
\end{remark}
\begin{theorem}\label{theorem_saddle} Let $h_\lambda(\bx)$ be continuous function of $\bx$ and $\cal U$ be the set such that for any point $\bx\in \cal U$ there exists a time $t(\bx)$ for which $\bs_t(\bx)\in {\cal U}_\epsilon$, the $\epsilon$ neighborhood of the origin equilibrium point. Let $h_\lambda(\bx)$ be the value of eigenfunction at point $\bx\in \cal U$, then 
\begin{align}
\left|h_\lambda(\bx)-\int_0^{t(\bx)} e^{-\lambda \tau}\bw_\lambda^\top \bF_n(\bs_t(\bx))d\tau\right|\leq \delta(\epsilon)
\end{align}
for some $\delta(\epsilon)>0$ and all $\bx\in \cal U$.
\end{theorem}
\begin{proof}
Using result of Theorem \ref{theorem_main}, we know $h_\lambda(\bx)$ satisfies (\ref{pde_soultion}). Consider any point $\bx\in \cal U$, and let $\by=\bs_{t(\bx)}(\bx)\in {\cal U}_\epsilon$. Subsuiting $\by=\bs_{t(\bx)}(\bx)$ in (\ref{pde_soultion}), we obtain
\[h_\lambda(\bx)-\int_0^{t(\bx)}e^{-\lambda \tau}\bw_\lambda^\top \bF_n(\bs_\tau(\bx))d\tau=e^{-\lambda t(\bx)}h_\lambda(\by)\]
Since $\lambda>0$, we have 
\[\left|h_\lambda(\bx)-\int_0^{t(\bx)}e^{-\lambda \tau}\bw_\lambda^\top \bF_n(\bs_\tau(\bx))d\tau\right|\leq |h_\lambda(\by)|\]
The proof then follows by using the contuinity property of $h_\lambda(\bx)$ as $\|\by\|\leq \epsilon$. \ 
\end{proof}
Using the results of the above theorem, we can approximate the eigenfunction for any point $\bx\in \cal U$ containing the stable manifold as 
\begin{align}
\hat \phi_\lambda(\bx)=\bw^\top_\lambda \bx+\int_0^{t(\bx)}e^{-\lambda \tau}\bw^\top_\lambda \bF_n(\bs_\tau(\bx))d\tau.\label{pathintegral}
\end{align}



\subsection{Computational Framework for Stable Manifold}

We first outline the algorithm to compute a stable manifold of the unstable equilibrium point and then use these steps to obtain the simulation results. \\

\noindent{\bf Algorithm to determine stability boundary}\\
\noindent 1)  Find all the equilibrium points and determine the equilibrium points, $\bx_i^\star$, that are unstable on the stability boundary.\\
\noindent 2) Let $\bA_i=\frac{\partial \bff}{\partial \bx}(\bx_i^\star)$ and  $\bw_\lambda$ be the left eigenvector corresponding to the eigenvalue $\lambda$ with ${\rm Re}(\lambda)>0$. Construct,

\[{\cal U}_0=\{\bx\in \bX: \bx^\top \bP\bx \leq \epsilon_1\}\]
for some positive matrix $\bP$ whose $n-1$ eigenvectors span the subspace $\{\bx: \bw_\lambda^\top \bx=0\}$. We will comment on the choice of $\epsilon_1$ and other parameters used in the algorithm later.\\
\noindent 3) Propagate the set ${\cal U}_0$ backward in time under the flow $\bs_{-t}(\bx)$ over time interval $[0,T]$ i.e.,
{\small \[{\cal U}=\bigcup_{t=0}^{T}\bs_{-t}({\cal U}_0)\]}
\noindent 4) Let $\{\bx_k\}_{k=1}^L$ be the sample data points uniformly distributed in the set ${\cal U}$. Use the path integral formula (\ref{pathintegral}) to compute the value of the eigenfunction corresponding to $\lambda$ at these points. 

\noindent 5) Let $\Psi(\bx)=(\psi_1(\bx),\ldots, \psi_N(\bx))^\top$ be the finite choice of basis functions used for the approximation of eigenfunction in the domain $\cal U$. Formulate the following least square optimization problem to determine the approximation of the eigenfunction in domain $\cal U$. 
\begin{align}\min_{\bu\in \mR^N}\|\bG\bu-\bc\|\label{leastsquare}
\end{align}
\begin{align*}\bc=(\hat\phi_\lambda(\bx_1),\ldots, \hat \phi_\lambda(\bx_L))^\top,\;\bG=(\Psi(\bx_1),\ldots, \Psi(\bx_L))^\top.
\end{align*}
The finite-dimensional approximation of the eigenfunction in the domain $\cal U$ is given by
\begin{align} \tilde\phi_\lambda(\bx)=\Psi(\bx)^\top\bu^\star,\;\;\;\bu^\star=\bG^\dagger \bc\label{optimalsol}
\end{align}

\noindent 6) The approximation of the stability boundary can then be identified as the zero-level curve of $\hat \phi_\lambda(\bx)$. In particular, we can determine the approximate stability boundary as $\hat \cS_b=\{\bx: |\tilde \phi_\lambda(\bx)|\leq \gamma\}$
for some small $\gamma>0$.

\begin{remark}The choice of $\epsilon_1$ will depend upon the total time used for simulation in backward time and the stable eigenvalue with the smallest real part, $\bar \lambda$. The $\epsilon_1$ neighborhood around the unstable equilibrium point will roughly grow in size as $e^{-\bar \lambda T}\epsilon_1$. The choice of $\epsilon_1$ and $T$ are connected to cover the stable manifold. The choice of $\gamma$ depends on how coarsely we can resolve zero. With the value of eigenfunction computed at different points in the state space, the $\gamma$ will be determined based on the smallest positive and largest negative value of the eigenfunctions at the sampled data.
\end{remark}


We provide results on the convergence analysis for the finite-dimensional approximation of the stable manifold. We start with the following assumption on the basis functions.
\begin{assumption}\label{assume_basis} We assume that the basis function, $\Psi(\bx)=(\psi_1(\bx),\ldots, \psi_N(\bx))^\top$, satisfies 
$\mu$ independence property, 
\begin{align}
\mu\{\bx\in \mR^n: \bc^\top\Psi(\bx)=0\}=0,
\end{align}
for all non-zero $\bc\in \mR^N$. The measure $\mu$ is assumed to be equivalent to Lebesgue measure, $m$, on $\bX$ i.e., $\mu(A)=0$ if and only if $m(A)=0$ for any measurable set $A\subset \bX$. The initial conditions $\{\bx_\ell\}_{\ell=1}^L$ are assumed to be independent identically distributed (i.i.d) sampled from  $\mu$. 
\end{assumption}
\noindent Under Assumption \ref{assume_basis}, it can be shown that the matrix $\bG$ in Eq. (\ref{leastsquare}) is invertible. However, due to space constraints we have to omit the proof. Let  $\cF_N\subset \cL_2(\bX,\mu)$ a finite-dimensional subspace spanned by
$\Psi(\bx)=[\psi_1(\bx),\ldots,\psi_N(\bx)]^\top$.
For $\phi\in \cL_2(\bX,\mu)$, we can define the projection of $\phi$ on the closed subspace $\cF_N\subset \cL_2(\bX,\mu)$ as  
\begin{align*}P_N^\mu\phi&=\arg\min_{f\in \cF_N}\|f-\phi\|_{\cL_2(\bX,\mu)}\nonumber\\&=\Psi^\top(\bx)\arg\min_{\bc\in \mR^n}\int_\bX |\bc^\top \Psi(\bx)-\phi(\bx)|^2 d\mu(\bx).
\end{align*}
Let $\hat \mu_L$ be the empirical measure i.e., $\hat \mu_L=\frac{1}{L}\sum_k^L \delta_{\bz_k}$, with $\delta_{\bz_k}$ the Dirac-delta measure. We then have
\begin{align*}P_N^{\hat \mu_L}\phi&=\arg\min_{f\in \cF_N}\|f-\phi\|_{\cL_2(\bX,\hat\mu_L)}\nonumber\\&=\Psi^\top(\bx)\arg\min_{\bc\in \mR^N}\frac{1}{L}\sum_{i=1}^M\left(\bc^\top \Psi(\bx_i)-\phi(\bx_i)\right)^2.\label{projefinite}
\end{align*}
\begin{lemma} 
 We have 
 \begin{align}\tilde \phi_\lambda(\bx) \to P_N^\mu \hat \phi_\lambda(\bx)\;\;{\rm as}\;\; L\to \infty
  \end{align}
\end{lemma}
\begin{proof} From (\ref{leastsquare}), (\ref{optimalsol}), and (\ref{projefinite}), it follows that $\tilde \phi_\lambda=P_N^{\hat \mu_L}\hat \phi_\lambda$. Since the sample $\{\bx_k\}$ are assumed to be drawn i.i.d. it follows from law or large numbers that $P_N^{\hat \mu_L}\hat \phi_\lambda \to P_N^{\mu}\hat \phi_\lambda$ as $L\to \infty$. 
\end{proof}

\section{Simulation Results}\label{section_simulation}

\subsection{Bistable Toggle Switch}
In the first example, we consider the dynamics of a genetic toggle switch. The basic dynamics comprise two biological states, representing proteins repressing each other to apply mutual negative feedback \cite{nandanoori2022data}.
{\small \begin{align}
\dot x_1= \frac{\alpha_1 }{1+x_2^{\beta_1}} - \eta_1 x_1,\;\;
\dot x_2=  \frac{\alpha_2 }{1+x_1^{\beta_2}} - \eta_2 x_2. \label{eq:tg_switch_dyn}
\end{align}}
Here, for the device parameters $\alpha_1,\alpha_2=1, \beta_1=3.55,  \beta_2=3.53,$ and $\eta_1,\eta_2= 0.5,$ (\ref{eq:tg_switch_dyn}) have two stable equilibrium points and one saddle point equilibrium at $(1,1)$. 

\begin{figure}[H]
    \centering
\includegraphics[scale=0.5]{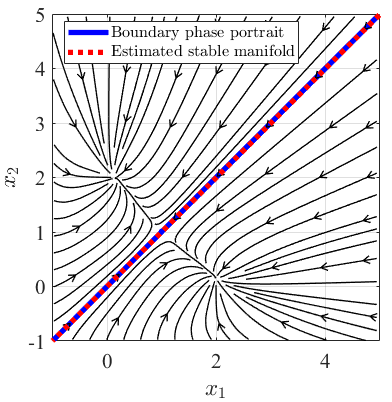}
    \caption{Stability boundary obtained from phase portrait (\textit{blue}) and estimated stable manifold of unstable equilibrium point (\textit{red}).}
    \label{fig:t_switch}
\end{figure}

The linearization of the system at the origin has eigenvalues $(0.3850,-1.3850)$.
The left eigenvector corresponding to the unstable eigenvalue is given by $\bw_{\lambda}=( 0.7061,-0.7081)^\top$.
It is clear from the phase portrait,  Fig. \ref{fig:t_switch}  (\textit{blue} curve), that the stability boundary is almost linear for this example. Therefore, we approximate the eigenfunction with linear only terms, i.e., $\bw_{\lambda}^\top \bx$. The corresponding zero-level curve of linear only eigenfunction is shown in Fig. \ref{fig:t_switch} (\textit{red} curve). From the figure, it is evident that the zero-level curve of eigenfunction matches well with the stability boundary.    


\subsection{2-D Speed Control System}
We consider a nonlinear speed control system as follows,
{\small
\begin{align}
    \dot x_1 = x_2,\;\;
    \dot x_2 = -K_d x_2 -x_1 -g x_1^2 \left(  \frac{x_2}{K_d} + x_1 +1\right)
    \label{eq:speed_control}
\end{align}}
Here, $K_d =1$ and $g=6$, the system has three equilibrium points, $(-0.7886, 0)$ and $ (0,0)$ are stable equilibrium and $(-0.21135, 0)$ is saddle (type-1) equilibrium point. In this example, we are interested in finding the stability boundary of $(0,0)$. The stable manifold of $(-0.21135, 0)$, characterized by a zero level set of the unstable eigenfunction of (\ref{eq:speed_control}), provides the stability boundary for $(0,0)$.  The linearization of the system around $(-0.21135, 0)$ gives eigenvalues $(0.4309,
   -1.6990)$. The eigenfunction corresponding to unstable eigenvalue $0.4309$ is the eigenfunction of interest.
   \begin{figure}[t]
    \centering
\includegraphics[scale=0.55]{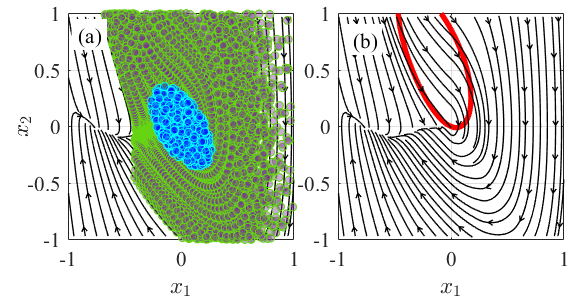}
    \caption{Speed control nonlinear system (a) Set ${\cal U}_0$ (\textit{blue}) and, set $\cal U$ obtained by backward reachable set (\textit{green}) (b) Estimated stable manifold (stability boundary) (\textit{red}).}
\label{fig:Data_speed_control}
\end{figure}
We consider a small region around $(-0.21135, 0)$ representing ${\cal U}_0$ with $500$ points  with $\epsilon_1 = 0.2$ as shown in Fig. \ref{fig:Data_speed_control}(a) (\textit{blue}) and obtain the set $\cal U$ by backward propagation of (\ref{eq:speed_control}) over time $[0, 10]$ as illustrated in Fig. \ref{fig:Data_speed_control}(a) (\textit{green}).
We restrict the data points in ${\cal U}$ in domain $[-1, 1]^2$ to compute the eigenfunction values using (\ref{pathintegral}) for $15850$ data points. In Fig. \ref{fig:Data_speed_control}(b), we show the estimated stability boundary on top of the phase portrait. The parameter  $\gamma = 5e^{-5}$. Clearly, the stability boundary follows the phase portrait representing the efficacy of the proposed approach.

\subsection{Two Generator Infinite Bus Power System}
The power system is a classic example of interest to engineering, where determining stability boundary is crucial. For the power system operator, it is of interest to know the {\it critical clearing time} (CCT), which is defined as the maximum allowable time to clear the fault without causing instability. The CCT is the time taken by the trajectory of the on-fault unstable power system dynamics to cross the stability boundary of the post-fault stable power system dynamics. The transient stability analysis of the power system to determine the CCT is then effectively reduced to computing the stability boundary of the post-fault stable power system dynamics. For this example, we consider a three-generator system \cite{chiang2011direct}, with generator three as the reference bus as follows. 
{\small \begin{align}
    \dot \delta_1 &= \omega_1, \; \dot \omega_1 = -\alpha_{1}\sin\delta_1 -\beta_{1}\sin(\delta_1 -\delta_2)-D_1\omega_1 \label{eq:power}\\
     \dot \delta_2 &= \omega_2, \; \dot \omega_2 = -\alpha_{2}\sin\delta_2 -\beta_{2} \sin(\delta_2 -\delta_1)-D_2\omega_1+P_m. \nonumber 
\end{align}}
Here, $\delta_i$ and $\omega_i$ are the generator rotor angle and angular velocity, respectively, for $i^{th}$ generator, $\alpha_i, \beta_i$ are the ratio of generator internal voltage and line impedance with $\alpha_{1}=1$, $\alpha_{2},\beta_{1},\beta_{2}=0.5$, $D_1=0.4$, and $D_2=0.5$.    The system has $6$ type-1 equilibrium points with $(\delta_1,\delta_2)$ values as $(3.24,0.31)$, $(3.04,3.24)$, $(0.03,3.10)$, $(-3.03,0.31)$, $(-3.24,-3.03)$, $(0.03,-3.17)$,
and the post-fault stable equilibrium point is $(0.02,0.06)$ with $\omega_1,\omega_2=0$. We consider ${\cal U}_0$ with $500$ data points and $\epsilon_1 =0.1$ and propagate this backward in time for $5$ sec. to contain $\cal U$. We compute the stable manifold corresponding to each unstable equilibrium point (UEP) as a zero-level set of the unstable eigenfunction as shown in  Fig. \ref{fig:manifols_power} (green)  with a combination of linear and trigonometric basis. The parameter $\gamma=1e^{-3}$. 
\begin{figure}[h]
    \centering    \includegraphics[scale=0.55]{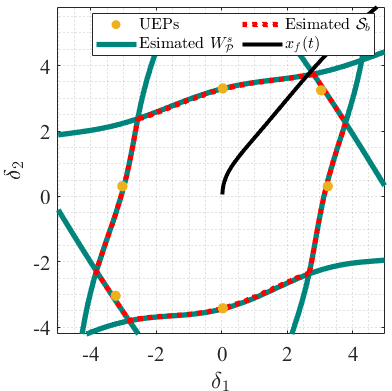}
    \caption{Estimated stability boundary projected on $\delta_1,\delta_2$ plane.}
\label{fig:manifols_power}
\end{figure}
 The joint level set of all the stable manifolds represents the stability boundary shown in  Fig. \ref{fig:manifols_power} (\textit{dotted red}).
 We consider an on-fault case by changing the parameters in (\ref{eq:power}) as 
 $\alpha_{1},\alpha_{2}=0.01$, $\beta{1}= 0.05$, and $\beta_{2}=0.001$
 shown in  Fig. \ref{fig:manifols_power} (\textit{black}). The CCT for this case using time domain simulation is $43.8$sec, and the CCT obtained from our approach, i.e., the time at which on-fault trajectory first crosses the stable manifold, is $43.7$sec. Therefore, this example shows the efficacy of our approach for practical applications of transient stability analysis in power systems.

\section{Conclusions and Discussion}\label{section_conclude}
We have proposed a direct approach for the stability analysis of a dynamical system using the Koopman spectrum. Compared to the existing approach based on the Lyapunov function, our proposed approach is more direct as it relies on the explicit computation of stability boundary using the zero-level curve of Koopman eigenfunction. We also provide a path-integral formula for the calculation of Koopman eigenfunction. One of the challenges with the proposed approach is that it requires the knowledge of unstable equilibrium points on the stability boundary.


\bibliographystyle{IEEEtran}
\bibliography{ref}

\end{document}